\documentclass[aps,pra,preprint,superscriptaddress,longbibliography,nofootinbib]{revtex4-1}
\usepackage[utf8]{inputenc}
\usepackage{amsfonts}
\usepackage{amsmath}
\usepackage{amssymb}
\usepackage{amsthm}
\usepackage{mathtools}
\usepackage{enumerate}
\usepackage[shortlabels]{enumitem}
\usepackage{array}
\usepackage{bm}
\usepackage{graphicx}
\usepackage{nicefrac}
\usepackage{booktabs}
\usepackage[all]{xy}
\usepackage{hyperref}
\usepackage{subcaption}
\usepackage{bbm}
\usepackage{tikz}
\usetikzlibrary{arrows}
\usetikzlibrary{positioning}
\usepackage{cancel}
\usepackage{xcolor}
\usepackage[font=small,labelfont=bf]{caption}
\usepackage{MnSymbol}
\usepackage{multirow}

\newcommand\mc[1]{\mathcal{#1}}
\newcommand\cH{\mc{H}}
\newcommand\cK{\mc{K}}
\newcommand\VH{\mc{V}(\cH)}

\newcommand\ra{\rightarrow}

\newtheorem{theorem}{Theorem}

\newtheorem{definition}{Definition}

\begin{document}

\title{From no-signalling to quantum states}
\author{Markus Frembs}
\email{m.frembs@griffith.edu.au}
\affiliation{Centre for Quantum Dynamics, Griffith University,\\ Yugambeh Country, Gold Coast, QLD 4222, Australia}
\vspace{-0.5cm}
\author{Andreas D\"oring}
\email{andreas.doering@posteo.de}

\begin{abstract}
    Characterising quantum correlations from physical principles is a central problem in the field of quantum information theory. Entanglement breaks bounds on correlations put by Bell's theorem, thus challenging the notion of \emph{local causality} as a physical principle. A natural relaxation is to study \emph{no-signalling} as a constraint on joint probability distributions. 
    It has been shown that when considered with respect to so-called locally quantum observables, bipartite non-signalling correlations never exceed their quantum counterparts; still, such correlations generally do not derive from quantum states. 
    This leaves open the search for additional principles which identify quantum states within the larger set of (collections of) non-signalling joint probability distributions over locally quantum observables. Here, we suggest a natural generalisation of no-signalling in the form of \emph{no-disturbance} to dilated systems. We prove that non-signalling joint probability distributions satisfying this extension correspond with bipartite quantum states up to a choice of time orientation in subsystems.
\end{abstract}

\maketitle

\section{Introduction} 
Given the violation of Bell inequalities \cite{AspectEtAl1981,ShalmEtAl2015,ZeilingerEtAl2015} (see also, \cite{WuShaknov1950}) and the paramount importance of entanglement for quantum information theory \cite{Shor1994,RaussendorfBriegel2001,AcinGisinMasanes2006,WisemanJonesDoherty2007,PawlowskiBrunner2011,deVicente2014,BravyiGossetKoenig2020,google2020},
characterising quantum correlations constitutes an important and ongoing research objective \cite{PawlowskiEtAl2009,CabelloSeveriniWinter2010,vanDam2013,FritzEtAl2013,JanottaHinrichsen2014,Popescu2014}. One line of this research has focused on the fact that while quantum mechanics is hardly reconcilable with the classical notion of \emph{local causality} \cite{Bell1976},
quantum correlations obey the more general principle of \emph{no-signalling}: 
the joint probability distributions for different local measurements $a,b$ with respective outcome sets $\{A\}$, $\{B\}$ marginalise to the same local distributions,
\begin{align}\label{eq: no-signalling probability distributions}
    \mu(A \mid a) = \sum_B \mu(A,B \mid a,b) \quad \quad \quad
    \mu(B \mid b) = \sum_A \mu(A,B \mid a,b)\; .
\end{align}
Importantly, Eq.~(\ref{eq: no-signalling probability distributions}) depends on the set of local measurements, consequently one must specify the possible choices of measurements on either subsystem in order to evaluate the constraints inherent to no-signalling. For instance, the PR-box correlations (for the CHSH scenario) restrict to just two measurements on either side \cite{PopescuRohrlich1994}. A physically more interesting scenario is that in which arbitrary local measurements are allowed. Remarkably, one can show that this already bounds the correlations to be quantum in the bipartite case \cite{BarnumEtAl2010,AcinEtAl2010}. Nevertheless, the collections of non-signalling distributions over the set of locally quantum observables do not correspond with quantum states \cite{KlayRandallFoulis1987,Wallach2002}. There are more non-signalling distributions than quantum states, i.e., while the correlations are as strong as quantum ones, the underlying distributions need not derive from a quantum state. 

In this paper we identify two physical principles that characterise those non-signalling bipartite correlations which correspond with quantum states. To set the stage, we review some basic facts about correlations over non-signalling correlations in Sec.~\ref{sec: No-signalling and locally quantum observables}. In Sec.~\ref{sec: Contextuality and no-disturbance} we reformulate the problem from the perspective of contextuality and identify \emph{no-disturbance} as the key underlying principle. This subsumes no-signalling, but makes explicit the intimate relationship with non-contextuality. Based on this, we suggest an extension of the no-disturbance principle to \emph{dilated systems} in Sec.~\ref{sec: no-disturbance and dilations}. Thm.~\ref{thm: from dilations to Jordan *-homomorphisms} shows that this strengthened principle almost singles out quantum states. In Sec.~\ref{sec: local unitary evolution} we identify the missing piece of data by introducing a notion of \emph{time orientation}. Our main result, Thm.~\ref{thm: main result}, proves that under a related consistency condition with respect to \emph{unitary evolution in subsystems} correlations indeed derive from quantum states. 
Sec.~\ref{sec: Conclusion and outlook} concludes.

\section{No-signalling and locally quantum observables}\label{sec: No-signalling and locally quantum observables}

Throughout, we denote by $\mc{L}(\cH)$ the algebra of linear operators on some finite-dimensional Hilbert space $\cH$, by $\mc{P}(\cH)$ the lattice of projections on $\cH$, and by $\mc{L}(\mc{H})_\mathrm{sa}$ the real-linear space of self-adjoint (Hermitian) operators, representing the observables of a system.

In contrast to restricted sets of observables such as those in \cite{PopescuRohrlich1994}, \cite{BarnumEtAl2010} study no-signalling for all \emph{locally quantum} observables. Locally quantum here means that every local system is described in terms of an observable algebra $\mc{L}(\cH)_\mathrm{sa}$ of self-adjoint (Hermitian) operators corresponding to the respective quantum system described by the Hilbert space $\cH$. However, rather than assuming the tensor product structure between Hilbert spaces, the composite system is described solely in terms of product observables, i.e., pairs $(a_1,a_2)$ where $a_i \in \mc{L}(\cH_i)_\mathrm{sa}$ for $i = 1,2$. Using the spectral decomposition of observables $a_i = \sum_j A_j q_i^j \in \mc{L}(\cH_i)_\mathrm{sa}$ for $A_j \in \mathbb{R}$ and $q^j_i \in \mc{P}(\cH_i)$, bipartite correlations arise from measures $\mu: \mc{P}(\cH_1) \times \mc{P}(\cH_2) \ra [0,1]$, $\mu(\mathbbm{1}) = \mu(\mathbbm{1}_1,\mathbbm{1}_2) = 1$ (here, $\mathbbm{1}_i \in \mc{L}(\cH_i)$ denotes the identity matrix). For the latter the no-signalling constraints in Eq.~(\ref{eq: no-signalling probability distributions}) read:
\begin{equation}\label{eq: no-signalling for measures}
    \mu(q_1) := \mu(q_1,\mathbbm{1}_2) = \sum_j \mu(q_1,q_2^j) \quad \quad \quad
    \mu(q_2) := \mu(\mathbbm{1}_1,q_2) = \sum_j \mu(q_1^j,q_2)
\end{equation}
for all mutually orthogonal sets of projections $(q_i^j)_j$, i.e., $\sum_j q^j_i = \mathbbm{1}_i$ and $q^j_iq^k_i = \delta_{jk}$.\footnote{For convenience, we restrict ourselves to (measures over) projective measurements. It is straightforward to extend the discussion to positive operator-valued measures (POVM), equivalently effect spaces. This allows to include the two-dimensional case in Thm.~\ref{thm: KlayRandallFoulis} (see \cite{BarnumEtAl2010,FrembsDoering2019b}).} 
For later reference, we point out that the structure of observables implies that (in addition to no-signalling) $\mu$ is also independent of contexts, in a sense to be made precise in Sec.~\ref{sec: Contextuality and no-disturbance}.

Applying Gleason's theorem \cite{Gleason1975} to the respective subsystems one shows that $\mu$ extends to a linear functional $\sigma_\mu: \mc{L}(\cH_1\otimes\cH_2) \ra \mathbb{C}$, equivalently, that there exists a linear operator $\rho_\mu$ such that $\sigma_\mu(a) = 
\mathrm{tr}[\rho_\mu a]$ for every $a 
\in \mc{L}(\cH_1\otimes\cH_2)$. 
Since $\mu$ is positive, $\sigma_\mu$ is positive on all product observables $a_1 \otimes a_2$ with $a_1 \in \mc{L}(\cH_1)_+$ and $a_2 \in \mc{L}(\cH_2)_+$. We call such normalised, i.e., $\mathrm{tr}[\rho_\mu]=1$, linear functionals \emph{positive on pure tensor (POPT) states}.

\begin{theorem}\label{thm: KlayRandallFoulis}
    \cite{KlayRandallFoulis1987,BarnumEtAl2010}
    Let $\cH_1$, $\cH_2$ be Hilbert spaces with $\mathrm{dim}(\cH_1),\mathrm{dim}(\cH_2) \geq 3$ finite. There is a one-to-one correspondence between measures $\mu: \mc{P}(\cH_1) \times \mc{P}(\cH_2) \ra [0,1]$, which satisfy the no-signalling constraints in Eq.~(\ref{eq: no-signalling for measures}), and POPT states $\sigma_\mu: \mc{L}(\cH_1\otimes\cH_2) \ra \mathbb{C}$. 
\end{theorem}

Crucially, a POPT state $\sigma_\mu$ is generally not positive. Hence, $\mu$ defines a quantum state if and only if $\sigma_\mu$ is positive. Despite this discrepancy, \cite{BarnumEtAl2010,AcinEtAl2010} show that the correlations arising from POPT states do not exceed those arising from quantum states in the bipartite case.

The existence of unphysical POPT states gives rise to the problem of finding a sharper classification that rules out such states.
This paper offers a solution to this problem. To this end, we first argue that rather than no-signalling, the natural principle underlying Eq.~(\ref{eq: no-signalling for measures}) is no-disturbance. We then extend the scope of the no-disturbance principle to \emph{dilations of local systems}, and show that up to a consistency condition with respect to \emph{unitary evolution in subsystems} this establishes a one-to-one correspondence between non-signalling joint probability distributions and bipartite quantum states.

We remark that distributions over product observables have been studied before, e.g. in \cite{KlayRandallFoulis1987}, out of the attempt to define a tensor product intrinsic to quantum logic.\footnote{For a recent contribution to this problem, see \cite{AbramskyBarbosa2021} (and references therein).} We do not consider this problem here, but note that our result has the potential to give new impulse to this research programme. Another perspective on this problem is given by Wallach in \cite{Wallach2002}, who considers a generalisation of Gleason's theorem to composite systems. Since it distracts from the physical significance of the present discussion, we study this problem in more generality elsewhere \cite{FrembsDoering2022b}.

\section{Non-contextuality and no-disturbance}\label{sec: Contextuality and no-disturbance}

A crucial feature of locally quantum non-signalling joint probability distributions is that they satisfy a more general constraint than no-signalling, called no-disturbance. We introduce this notion in the following sections, before extending the no-disturbance principle to dilations in Sec.~\ref{sec: no-disturbance and dilations}. Our approach naturally embeds within the analysis of the principal role of contextuality in quantum theory \cite{FreDoe19a}. For a recent review of the wider subject, see \cite{Masse2021}.

\subsection{Non-contextuality and marginalisation constraints}

Contextuality is a key principle in foundations, separating quantum from classical physics \cite{Specker1960,KochenSpecker1967,IshamButterfieldI,DoeringIsham2011,AbramskyBrandenburger2011,CabelloSeveriniWinter2010}. Moreover, it has been shown to be a resource for quantum computation in various architectures \cite{Howard2014,Raussendorf2016,BravyiGossetKoenig2020}.

At the core of contextuality lies the notion of \textit{simultaneous measurability}, which equips the set of observables with a reflexive, symmetric, but generally non-transitive relation \cite{Specker1960}. We call any subset of simultaneously measurable observables a \textit{context}. The set of all contexts carries an intrinsic order relation arising from \textit{coarse-graining} on outcomes of observables. The resulting partially ordered set is called the \textit{partial order of contexts}. For the close connection between contextuality in this sense and the non-existence of valuation functions in the sense of Kochen and Specker \cite{KochenSpecker1967}, we refer to \cite{FreDoe19a}.

In quantum theory, two observables are simultaneously measurable if and only if they commute. Consequently, contexts are given by commutative subalgebras $V \subseteq \mc{L}(\mc{H})$, which are ordered by inclusion into the corresponding partial order of contexts denoted by $\mc{V}(\mc{H})$. 
In this setup a quantum state becomes a collection of probability distributions $(\mu_V)_{V \in \mc{V}(\mc{H})}$, one for every context. Moreover, \textit{coarse-graining} constrains these \textit{across} contexts: let $\mu_{\tilde{V}}$, $\mu_V$ be probability distributions in contexts $\tilde{V},V \in \VH$ such that $\tilde{V} \subset V$, and denote by $i_{\tilde{V}V}: 
\mc{P}(\tilde{V}) \hookrightarrow \mc{P}(V)$ the inclusion relation between their respective projections, then $\mu_{\tilde{V}}$ is obtained from $\mu_V$ by marginalisation,
\begin{equation}\label{eq: marginalisation}
    \forall q \in \mc{P}(\tilde{V}):\quad \mu_{\tilde{V}}(q) = \mu_V(i_{\tilde{V}V}(q)) = \mu_V|_{\tilde{V}}(q)\; .
\end{equation}
While individually inconspicuous, marginalisation constraints impose a strong condition in conjunction with \emph{non-contextuality}: probabilities of events corresponding to an observable $a \in \tilde{V} \subset V,V'$ are independent of other observables $b \in V$, $c \in V'$, i.e., they are independent of context:
\begin{equation}\label{eq: non-contextuality}
    \forall \tilde{V},V,V' \in\mc{V}(\cH), \tilde{V} \subset V,V': \quad \mu_{V'}|_{\tilde{V}} = \mu_{\tilde{V}} = \mu_V|_{\tilde{V}}\; .
\end{equation}
This non-contextuality assumption on probability distributions $\mu_V$ is at the heart of Gleason's theorem \cite{Gleason1975}; for a reformulation of the latter in this language 
see \cite{Doering2008,FreDoe19a,FrembsDoering2022b}.

\subsection{No-signalling from no-disturbance}

There is a natural notion of composition for (partial orders of) contexts: their canonical product, denoted $\mc{V}(\cH_1) \times \mc{V}(\cH_2)$, is the Cartesian product with elements $(V_1,V_2)$ for $V_1 \in \mc{V}(\cH_1)$, $V_2 \in \mc{V}(\cH_2)$ and order relations such that for all $\tilde{V}_1,V_1 \in \mc{V}(\cH_1)$, $\tilde{V}_2,V_2 \in \mc{V}(\cH_2)$,
\begin{equation}\label{eq: product context category}
    (\tilde{V}_1,\tilde{V}_2) \subseteq (V_1,V_2) :\Longleftrightarrow \tilde{V}_1 \subseteq_1 V_1 \mathrm{\ and \ } \tilde{V}_2 \subseteq_2 V_2\; .
\end{equation}
Restricted to product contexts, Eq.~(\ref{eq: non-contextuality}) says that for all $\tilde{V}_1 \subset V_1,V_1' \in\mc{V}(\cH_1)$ and $\tilde{V}_2 \subset V_2,V_2' \in\mc{V}(\cH_2)$
\begin{equation}\label{eq: no-disturbance}
     \mu_{(V_1, V_2)}|_{(\tilde{V}_1, V_2)} = \mu_{(\tilde{V}_1, V_2)} = \mu_{(V_1', V_2)}|_{(\tilde{V}_1, V_2)} \quad \quad \quad
     \mu_{(V_1, V_2)}|_{(V_1, \tilde{V}_2)} = \mu_{(V_1, \tilde{V}_2)} = \mu_{(V_1, V'_2)}|_{(V_1, \tilde{V}_2)}\; .
\end{equation}
We call the collection of constraints over all contexts $(V_1,V_2) \in \mc{V}(\cH_1) \times \mc{V}(\cH_2)$ in Eq.~(\ref{eq: no-disturbance}) the \emph{no-disturbance principle}.\footnote{The term appears in \cite{RamanathanEtAl2012}, but has been used under different names before, e.g. in \cite{CabelloSeveriniWinter2010} (see also, \cite{DoeringIsham2011,FreDoe19a}).}
Note that no-disturbance reduces to no-signalling in Eq.~(\ref{eq: no-signalling probability distributions}) when restricted to the trivial contexts $1_i := \mathbb{C}\mathbbm{1}_i \subset \mc{L}(\cH_i)$\footnote{The respective contexts correspond to the trivial events of observing that there is a local system.} on the respective local subsystem,
\begin{equation}\label{eq: no-signalling from no-disturbance}
    \mu_{(V_1, V_2)}|_{(V_1, 1_2)} = \mu_{(V_1,1_2)} = \mu_{V_1,V'_2}|_{(V_1,1_2)} \quad \quad \quad 
    \mu_{(V_1, V_2)}|_{(1_1,V_2)} = \mu_{(1_1,V_2)} = \mu_{(V'_1,V_2)}|_{(1_1,V_2)}\; . 
\end{equation}
Observe that general non-signalling joint probability distributions over locally quantum observables correspond with collections of probability distributions $(\mu_V)_{V \in \mc{V}(\cH_1) \times \mc{V}(\cH_2)}$ satisfying Eq.~(\ref{eq: no-signalling from no-disturbance}). The latter are more general than the non-signalling measures $\mu: \mc{P}(\cH_1) \times \mc{P}(\cH_2) \ra [0,1]$ in Eq.~(\ref{eq: no-signalling for measures}), which satisfy the more restrictive no-disturbance condition in Eq.~(\ref{eq: no-disturbance}). In this sense, no-disturbance is already implicit in \cite{BarnumEtAl2010}. With the latter, we re-emphasise that $\mc{V}(\cH_1) \times \mc{V}(\cH_2) \neq \mc{V}(\cH_1 \otimes \cH_2)$ is a meagre part of the full quantum structure of the composite Hilbert space $\cH_1 \otimes \cH_2$. In particular, the product system is not described by the tensor product of the individual Hilbert spaces.

Having established the crucial role of non-contextuality and no-disturbance, in the next section we extend this principle to dilations on local subsystems.

\section{No-disturbance for dilations}\label{sec: no-disturbance and dilations}

A key idea of this work is to impose the no-disturbance condition in Eq.~(\ref{eq: no-disturbance}) not only between locally quantum observables, but to extend its scope to dilations of (at least) one of the two subsystems. We give some motivational background for this step first.

Recall that by Thm.~\ref{thm: KlayRandallFoulis} a measure $\mu: \mc{P}(\cH_1) \times \mc{P}(\cH_2) \ra [0,1]$ defines a POPT state $\mu(q_1,q_2) = \sigma_\mu(q_1 \otimes q_2) = \mathrm{tr}[\rho_\mu(q_1\otimes q_2)]$ with corresponding linear operator $\rho_\mu \in \mc{L}(\cH_1 \otimes \cH_2)$. Similarly, for every $q_1 \in \mc{P}(\cH_1)$ we define a positive linear operator $\rho_\mu(q_1) \in \mc{L}(\cH_2)_+$ from 
$\mu(q_1,q_2) =: \mathrm{tr}_{\cH_2}[\rho_\mu(q_1)q_2]$ for all $q_2 \in \mc{P}(\cH_2)$. It follows from Thm.~\ref{thm: KlayRandallFoulis} that $\rho_\mu$ extends to a positive linear map $\phi_\mu: \mc{L}(\cH_1) \ra \mc{L}(\cH_2)$. On the other hand, for any single $q_1 \in \mc{P}(\cH_1)$ we may assume 
$\rho_\mu(q_1) \in \mc{L}(\cH_2)_+$ to arise via coarse-graining from some 
$|\psi_\mu(q_1)\rangle \in \cK := \cH_2 \otimes \cH_E$ on a larger system (to be thought of as the system together with an (environmental) ancillary system), by tracing out the extra degrees of freedom:
\begin{equation}\label{eq: purification}
    \rho_\mu(q_1)
    = \mathrm{tr}_{\cH_E}[|\psi_\mu(q_1)\rangle\langle\psi_\mu(q_1)|]
    = \mathrm{tr}_{\cH_E}[u^*(q_1 \otimes |\psi_\mu\rangle\langle\psi_\mu|)u]\; .
\end{equation}
$|\psi_\mu(q_1)\rangle$ is a \emph{purification} of 
$\rho_\mu(q_1)$,
where the unitary $u \in \mc{L}(\cH_2 \otimes \cH_E)$ is chosen to separate the input from the ancillary system. We want to arrange the purifications for all $q_1 \in \mc{P}(\cH_1)$ in an economical way. To start with, note that since $\rho_\mu(q_1) \in \mc{L}(\cH_2)_+$ for every $q_1 \in \mc{P}(\cH_1)$, $(\rho_\mu(q_1^i))_i$ defines a non-normalised ($\sum_i \rho_\mu(q_1^i) = \rho_\mu(\mathbbm{1}_1) = \mathrm{tr}_{\cH_1}[\rho_\mu]]$) positive operator-valued measure (POVM) on $\mc{L}(\cH_2)$ for every set of mutually orthogonal projections $(q_1^i)_i$ in $\mc{L}(\cH_1)$, i.e., for every context $V_1 \in \mc{V}(\cH_1)$. We may thus write
\begin{equation}\label{eq: dilation in context}
    \rho^{V_1}_\mu(q_1^i)
    = \mathrm{tr}_{\cH_E}\left[u^*\left(q_1^i \otimes |\psi_\mu^{V_1}\rangle\langle\psi_\mu^{V_1}|\right)u\right]\; ,
\end{equation}
for all $q_1^i \in \mc{P}(V_1)$ and $V_1 \in \mc{V}(\cH_1)$, mimicking the purification in Eq.~(\ref{eq: purification}). Eq.~(\ref{eq: dilation in context}) is called a \emph{dilation} of the POVM $(\rho_\mu(q_1^i))_i$ to the projection-valued measure (PVM) $q_1 \mapsto q_1 \otimes \mathbbm{1}_{\cH_E}$. More generally, by Naimark's theorem \cite{Naimark1943,Stinespring1955} every POVM $\rho^{V_1}_\mu$ admits a dilation to a PVM $\varphi_\mu^{V_1}: \mc{P}(V_1) \ra \mc{P}(\cK)$ such that $\rho_\mu = v^*\varphi_\mu^{V_1}v$, where $v: \cH_2 \ra \cK$ is a linear map.\footnote{Here, we concentrate the dependence on contexts in the mapping $V_1 \mapsto \varphi^{V_1}_\mu$, by choosing $\cK$ sufficiently large \cite{Naimark1943,Stinespring1955,Choi1975} and 
by absorbing any context dependence on $v_{V_1}$ into $\varphi_\mu^{V_1}$ for every $V_1 \in \mc{V}(\cH_1)$.} 

Importantly, the dilations $V_1 \mapsto \varphi^{V_1}_\mu$ generally depend on contexts $V_1 \in \mc{V}(\cH_1)$. In contrast, recall that no-disturbance in Eq.~(\ref{eq: no-disturbance}) encodes non-contextuality constraints between product observables. By comparison, this suggests to extend the no-disturbance principle also to product observables with respect to the dilations $(\varphi^{V_1}_\mu)_{V_1 \in \mc{V}(\cH_1)}$. 

\begin{definition}\label{def: no-disturbance for dilations}
    We say that the measure $\mu: \mc{P}(\cH_1) \times \mc{P}(\cH_2) \ra [0,1]$ satisfies \emph{no-disturbance for dilations} if $\mu(q_1,q_2) = \mathrm{tr}_{\cH_2}[\left(v^*\varphi_\mu^{V_1}(q_1)v\right)q_2]$ for some Hilbert space $\cK$, linear map $v: \cH_2 \ra \cK$,
    and projection-valued measures $(\varphi^{V_1}_\mu)_{V_1 \in \mc{V}(\cH_1)}$, $\varphi^{V_1}_\mu: \mc{P}(V_1)\ra \mc{P}(\cK)$ such that
    \begin{equation}\label{eq: no-disturbance for dilations}
        \forall q_1 \in \mc{P}(V_1),V_1 \in \mc{V}(\cH_1), q_2' \in \mc{P}(\cK): \quad \mu'(q_1,q_2')
        := \mathrm{tr}_{\cH_2}\left[v^*\varphi^{V_1}_\mu(q_1)q_2'v\right]
    \end{equation}
    satisfies the no-disturbance principle in Eq.~(\ref{eq: no-disturbance}) for all product contexts in $\mc{V}(\cH_1) \times \mc{V}(\cK)$.
\end{definition}

From a physical point of view, we interpret the (conditional) probability distributions $\mu_{V_2}(q_1,\cdot)$ in contexts $V_2 \in \mc{V}(\cH_2)$ as states of incomplete information. We have seen that this interpretation arises naturally from the explicit dilations in Eq.~(\ref{eq: dilation in context}), which express the state as arising from coarse-graining of ancillary degrees of freedom. This view is further corroborated if one allows not only projective measurements on the respective subsystems $\mc{L}(\cH_1)$ and $\mc{L}(\cH_2)$, but arbitrary positive operator-valued measures, e.g. as discussed in \cite{BarnumEtAl2010}.

Put the other way around, if $\mu$ did not satisfy no-disturbance for dilations in Def.~\ref{def: no-disturbance for dilations}---while potentially being non-disturbing (non-signalling) when restricted to the system $\cH_1 \otimes \cH_2$---it would fail to be non-disturbing (non-signalling) \emph{for every choice of dilations} $(\varphi^{V_1}_\mu)_{V_1 \in \mc{V}(\cH_1)}$ to a larger system. Such measures 
would at the very least prompt a substantial revision of the concept of mixed states in terms of states of information in quantum theory. As such, the extension to dilated systems in Def.~\ref{def: no-disturbance for dilations} is a natural one, and arguably conservative compared with other related approaches, e.g. \cite{PawlowskiEtAl2009, vanDam2013,FritzEtAl2013,Popescu2014}.

Note also that Def.~\ref{def: no-disturbance for dilations} does not change the fact that the composite system is described by means of the Cartesian product of contexts in Eq.~(\ref{eq: product context category}), not the tensor product.\\


In order to state the implications of Def.~\ref{def: no-disturbance for dilations}, we remind the reader of some basic facts about Jordan algebras. Recall that the set of self-adjoint (Hermitian) matrices $\mc{L}(\mc{H})_\mathrm{sa}$ is closed under the \textit{anti-commutator} $\{a,b\} = ab + ba$ for all $a,b \in \mc{L}(\mc{H})_\mathrm{sa}$. This defines the (special) Jordan algebra $\mc{J}(\cH) = (\mc{L}(\mc{H})_\mathrm{sa},\{\cdot,\cdot\})$.\footnote{A Jordan algebra is called \emph{special} if it is isomorphic to the subalgebra of the self-adjoint part of an associative algebra. Otherwise, it is called \emph{exceptional}.} Moreover, we obtain a Jordan $*$-algebra by extending the operation $\{\cdot,\cdot\}$ to the complexified algebra $\mc{J}(\cH) = \mc{J}(\mc{H})_\mathrm{sa} + i\mc{J}(\mc{H})_\mathrm{sa}$. This expresses the fact that $\mc{J}(\mc{H})_\mathrm{sa}$ is the self-adjoint part of $\mc{L}(\cH)$. Finally, a Jordan $*$-homomorphism $\Phi: \mc{J}(\cH_1) \rightarrow \mc{J}(\cH_2)$ is a linear map $\Phi: \mc{L}(\cH_1) \ra \mc{L}(\cH_2)$ such that $\Phi(\{a,b\}) = \{\Phi(a),\Phi(b)\}$ and $\Phi^*(a) = \Phi(a^*)$ for all $a,b \in \mc{J}(\cH_1)$.

We have the following key result, which extracts Jordan structure from Def.~\ref{def: no-disturbance for dilations}.

\begin{theorem}\label{thm: from dilations to Jordan *-homomorphisms}
    Let $\cH_1$, $\cH_2$ be Hilbert spaces with $\mathrm{dim}(\cH_1),\mathrm{dim}(\cH_1) \geq 3$ finite, and let $\mu: \mc{P}(\cH_1) \times \mc{P}(\cH_2) \ra [0,1]$ satisfy \emph{no-disturbance for dilations} in Def.~\ref{def: no-disturbance for dilations}. 
    Then the linear functional $\sigma_\mu: \mc{L}(\cH_1 \otimes \cH_2) \ra \mathbb{C}$ in Thm.~\ref{thm: KlayRandallFoulis} is of the form
    \begin{equation*}
         \sigma_\mu(a \otimes b)
         = \mathrm{tr}_{\cH_2}\left[\left(v^*\Phi_\mu(a)v\right)b\right]
    \end{equation*}
    for some 
    Hilbert space $\cK$, linear map $v:\cH_2 \ra \cK$, and Jordan $*$-homomorphism $\Phi_\mu: \mc{J}(\cH_1) \ra \mc{J}(\cK)$.
\end{theorem}

\begin{proof}
    Since $\mu$ satisfies no-disturbance for dilations in Def.~\ref{def: no-disturbance for dilations}, there exists a Hilbert space $\cK$, a linear map $v: \cH_2 \ra \cK$, and projection-valued measures $(\varphi^{V_1}_\mu)_{V_1 \in \mc{V}(\cH_1)}$ such that $\mu(q_1,q_2) = \mathrm{tr}_{\cH_2}\left[\left(v^*\varphi_\mu^{V_1}(q_1)v\right)q_2\right]$ for all $q_1 \in \mc{P}(V_1)$, $V_1 \in \mc{V}(\cH_1)$ and $q_2 \in \mc{P}(\cH_2)$. Moreover, the extension $\mu'(q_1,q_2') = \mathrm{tr}_{\cH_2}\left[v^*\varphi^{V_1}_\mu(q_1)q_2'v\right]$ in Eq.~(\ref{eq: no-disturbance for dilations}), where $q_1 \in \mc{P}(V_1)$, $V_1 \in \mc{V}(\cH_1)$, and $q_2' \in \mc{P}(\cK)$, satisfies the no-disturbance principle in Eq.~(\ref{eq: no-disturbance}). The latter are equivalent to the non-contextuality constraints in Eq.~(\ref{eq: non-contextuality}), restricted to product contexts $\mc{V}(\cH_1) \times \mc{V}(\cK)$. Consequently, $(\varphi^{V_1}_\mu)_{V_1 \in \mc{V}(\cH_1)}$ does not depend on contexts, and therefore defines a map $\varphi_\mu: \mc{P}(\cH_1) \ra \mc{P}(\cK)$. It follows that $\varphi_\mu$ is an orthomorphism, i.e., (i) $\varphi_\mu(0) = 0$, (ii) $\varphi_\mu(1-p) = 1-\varphi_\mu(p)$, (iii) $\varphi_\mu(p)\varphi_\mu(q)=0$ and (iv) $\varphi_\mu(p+q) = \varphi_\mu(p)+\varphi_\mu(q)$ all hold whenever $p,q \in \mc{P}(\cH_1)$, $pq=0$. This follows since by definition $\varphi^{V_1}_\mu: \mc{P}(\cH_1) \ra \mc{P}(\cK)$ is a projection-valued measure in every context $V_1 \in \mc{V}(\cH_1)$. In particular, note that conditions (i)-(iii) hold whenever $p,q \in \mc{P}(\cH_1)$, $pq=0$, which implies $p,q \in V_1$ for some $V_1 \in \mc{V}(\cH_1)$.
    Finally, by a result due to Bunce and Wright (Cor.~1, \cite{BunceWright1993}) every orthomorphism $\varphi: \mc{P}(\cH_1) \ra \mc{P}(\cK)$ lifts to a Jordan $*$-homomorphism $\Phi: \mc{J}(\cH_1) \ra \mc{J}(\cK)$ as desired.
\end{proof}

Of course, every quantum state has a purification which yields a dilation of the form in Thm.~\ref{thm: from dilations to Jordan *-homomorphisms}. 
Our argument works in the reverse direction: by requiring that measures $\mu$ have a non-contextual extension, i.e., that they satisfy the no-disturbance principle for at least one choice of dilations $V_1 \mapsto \varphi^{V_1}_\mu$, $\mu$ is of the form in Thm.~\ref{thm: from dilations to Jordan *-homomorphisms}. Next, we show that, in this case, $\mu$ already defines a quantum state up to a choice of time orientation in local subsystems.

\section{Unitary evolution and time-oriented states}\label{sec: local unitary evolution}

Extending no-disturbance to dilations via Def.~\ref{def: no-disturbance for dilations} is not quite sufficient to ensure that a POPT state $\sigma_\mu = \mathrm{tr}[\rho_\mu\ \cdot]$ becomes positive and thus a quantum state---but it almost is. The difference is best expressed in terms of the Choi-Jamio\l kowski isomorphism \cite{Jamiolkowski1972}.
In particular, by Choi's theorem the linear operator $\rho_\mu$ is positive if and only if a related map $\phi_\mu$ is completely positive \cite{Choi1975}.
Moreover, by Stinespring's theorem a map is completely positive if and only if it corresponds with an 
algebra homomorphism (a representation) on a larger system \cite{Stinespring1955}, i.e., $\phi_\mu: \mc{L}(\cH_1) \ra \mc{L}(\cH_2)$ is completely positive if and only if there exists a Hilbert space $\cK$, a linear map $v: \cH_2 \ra \cK$, and an algebra homomorphism $\Phi_\mu: \mc{L}(\cH_1) \ra \mc{L}(\cK)$ such that $\phi_\mu(a) = v^*\Phi_\mu(a)v$. 
By Thm.~\ref{thm: from dilations to Jordan *-homomorphisms}, if $\mu$ satisfies no-disturbance for dilations, it is of this form with the only difference that $\Phi_\mu$ is a merely Jordan $*$-homomorphism.

The Jordan algebra $\mc{J}(\mc{H})$ does not completely determine the 
algebra $\mc{L}(\mc{H})$, since it lacks the antisymmetric part or \textit{commutator} $[a,b] = ab - ba$ in the associative (and noncommutative for $\mathrm{dim}(\cH) \geq 2$) product $ab = \frac{1}{2}(ab + ba) + \frac{1}{2}(ab - ba)$ for all $a,b \in \mc{L}(\mc{H})$. In particular, the Jordan $*$-homomorphism $\Phi_\mu$ is an 
algebra homomorphism if and only if it also preserves commutators, i.e., $\Phi([a,b]) = [\Phi(a),\Phi(b)]$ for all $a,b \in \mc{J}(\cH)$. This property has a distinctive physical meaning in terms the unitary evolution in local subsystems.

To see this, we describe the dynamics of a system represented by the observable algebra $\mc{L}(\cH)_\mathrm{sa}$ in terms of one-parameter groups of automorphisms $\mathbb{R} \mapsto \mathrm{Aut}(\mc{L}(\cH)_\mathrm{sa})$.  Note that $\mc{L}(\cH)$ possesses a canonical action on itself by conjugation $\Psi: \mathbb{R} \times \mc{L}(\cH)_{\mathrm{sa}} \ra \mathrm{Aut}(\mc{L}(\cH)_\mathrm{sa})$, $\Psi(t,a)b = e^{ita}be^{-ita}$ for all $a,b \in \mc{L}(\cH)_\mathrm{sa}$.\footnote{The map $\Psi$ `selects' unitary (as opposed to anti-unitary \cite{Wigner_GruppentheorieUndQM,Bargmann1964}) symmetries on $\mc{J}(\cH)_\mathrm{sa}$. Inherent in this is an identification of self-adjoint elements in $\mc{J}(\cH)_\mathrm{sa}$ and generators of symmetries, called a \emph{dynamical correspondence}. In other words, $\Psi$ `selects' a canonical dynamical correspondence on $\mc{J}(\cH)_\mathrm{sa}$: $\Psi = \exp \circ \psi$ is the exponential of the (canonical) dynamical correspondence $\psi$ on $\mc{L}(\cH)$. For more details, see \cite{AlfsenShultz1998,Doering2014}.} The latter expresses the unitary evolution in the system, thus promoting the parameter $t$ to a time parameter with $a$ playing the role of a Hamiltonian. Note however that without fixing a preferred Hamiltonian the value of $t$ has a priori no objective physical meaning, it is intrinsically relational. Nevertheless, the sign of $t$ turns out to be independent of (the choice of Hamiltonian) $a \in \mc{L}(\cH)_\mathrm{sa}$ \cite{AlfsenShultz1998}.

By differentiation, $\left.\frac{d}{dt}\right\vert_{t=0}\Psi(t,a) = i[a,\cdot]$, where $[\cdot,\cdot]$ is the commutator in the ambient algebra $\mc{L}(\cH)$. It follows that $\Phi_\mu$ preserves commutators if and only if it preserves the canonical unitary evolution $\Psi$ of the subsystem algebras. We call $\Psi$ 
the \emph{canonical time orientation} on $\mc{L}(\cH)$ \cite{AlfsenShultz1998,Doering2014,Frembs2022a,Frembs2022b}, and $\Psi^*(t,a) = * \circ \Psi(t,a) = \Psi(-t,a)$ (cf. Prop.~15 in \cite{AlfsenShultz1998a}) the \emph{reverse time orientation}, which corresponds with the opposite order of composition in $\mc{L}(\cH)$, equivalently the opposite sign for the commutator in the respective algebras,
\begin{equation*}
    \mc{L}_\pm(\mc{H}) := (\mc{J}(\cH),\Psi) = (\mc{J}(\mc{H}), \cdot_\pm)\; ,\footnote{These are the only associative products which reduce to $\mc{J}(\mc{H})$ on the symmetric part \cite{AlfsenShultz1998,Kadison1951}.}
\end{equation*}
where we set $a\cdot_\pm b := \frac{1}{2}\{a,b\} \pm \frac{1}{2}[a,b]$ for all $a,b \in \mc{J}(\cH)$. Of course, $\mc{L}(\cH)_+ = \mc{L}(\cH)$.

\begin{definition}\label{def: time-oriented states}
    We say that a measure $\mu: \mc{P}(\cH_1) \times \mc{P}(\cH_2) \ra [0,1]$, which satisfies no-disturbance for dilations in Def.~\ref{def: no-disturbance for dilations}, is \emph{time-oriented with respect to $\mc{L}_-(\cH_1)$ and $\mc{L}(\cH_2)$} if
    \begin{equation}\label{eq: time-oriented states}
        \forall t \in \mathbb{R}, a \in \mc{J}(\mc{H}_1)_\mathrm{sa}:\ \Phi_\mu \circ \Psi^*_1(t,a)
        = \Phi_\mu \circ * \circ  \Psi_1(t,a)
        = \Psi'_2(t,\Phi_\mu(a)) \circ \Phi_\mu\; ,
    \end{equation}
    where $\Phi_\mu$ is the Jordan $*$-homomorphism in Thm.~\ref{thm: from dilations to Jordan *-homomorphisms}.
\end{definition}

Two remarks on Def.~\ref{def: time-oriented states} are in order. First, note that Eq.~(\ref{eq: time-oriented states}) requires consistency with respect to the commutators on $\mc{L}(\cH_1)$ and $\mc{L}(\cK)$. However, since $\mc{L}(\cH_2)$ arises from $\mc{L}(\cK)$ by restriction, the time orientation $\Psi_2$ 
on $\mc{L}(\cH_2)$ extends to a unique time orientation $\Psi'_2$ 
on $\mc{L}(\cK)$. Second, we have used the reverse time orientation $\Psi^*_1$ on the first system. This choice of relative time orientation $(\Psi_1^*,\Psi_2)$ is intrinsic to Choi's theorem \cite{Choi1975} (see also \cite{Frembs2022b}), which allows us to deduce positivity of $\sigma_\mu$ from complete positivity of $\phi_\mu$ in Thm.~\ref{thm: from dilations to Jordan *-homomorphisms}.

In fact, the conjunction of Def.~\ref{def: no-disturbance for dilations} and Def.~\ref{def: time-oriented states} yields our main result.

\begin{theorem}\label{thm: main result}
    Let $\cH_1$, $\cH_2$ be Hilbert spaces with $\mathrm{dim}(\cH_1),\mathrm{dim}(\cH_2) \geq 3$ finite, and let $\mu: \mc{P}(\cH_1) \times \mc{P}(\cH_2) \ra [0,1]$ be time-oriented as by Def.~\ref{def: time-oriented states} (and thus satisfy no-disturbance for dilations). Then $\mu$ uniquely extends to a quantum state $\sigma_\mu \in \mc{S}(\cH_1 \otimes \cH_2)$.
\end{theorem}

\begin{proof}
    Since $\mu$ satisfies no-disturbance for dilations in Def.~\ref{def: no-disturbance for dilations}, the dilations $(\varphi^{V_1}_\mu)_{V_1 \in \mc{V}(\cH_1)}$ uniquely extend to a Jordan $*$-homomorphism $\Phi_\mu$ by Thm.~\ref{thm: from dilations to Jordan *-homomorphisms}. In particular, $\Phi_\mu$ preserves anti-commutators. Moreover, since $\mu$ is time-oriented (see Def.~\ref{def: time-oriented states}), $\Phi_\mu$ also preserves commutators between the algebras $\mc{L}_-(\cH_1)$ and $\mc{L}(\cH_2)$, hence, is an 
    algebra homomorphism. By Stinespring's theorem, this implies that the linear map $\phi_\mu: \mc{L}_-(\cH_1) \ra \mc{L}(\cH_2)$ is completely positive. Finally, we define a linear operator $\rho_\mu \in \mc{L}(\cH_1 \otimes \cH_2)$ by the relation $\phi_\mu(a) = \mathrm{tr}_{\cH_1}[\rho_\mu(a \otimes \mathbbm{1}_2)]$ for all $a \in \mc{L}(\cH_1)$. Note that $\phi_\mu \circ *$ is the inverse of the Choi-Jamio\l kowski isomorphism \cite{Jamiolkowski1972,Choi1975}, defined on the computational basis $\{|i\rangle\}_i$ in $\cH_1$ by
    \begin{equation}\label{eq: CJI}
        \rho_\mu
        = \sum_{ij} |i\rangle\langle j| \otimes \phi_\mu(|j\rangle\langle i|)
        = \sum_{ij} |i\rangle\langle j| \otimes (\phi_\mu \circ *)(|i\rangle\langle j|)\; ,
    \end{equation}
    where $*$ denotes the Hermitian adjoint (see also \cite{Frembs2022b}). Since $\phi_\mu: \mc{L}_-(\cH_1) \ra \mc{L}(\cH_2)$ is completely positive and since $*$ reverses time orientations by Eq.~(\ref{eq: time-oriented states}), it follows from Choi's theorem \cite{Choi1975} that $\rho_\mu \in \mc{L}(\cH_1\otimes \cH_2)$ is positive. Hence, $\sigma_\mu = \mathrm{tr}[\rho_\mu\ \cdot]$ defines a quantum state. By construction, $\sigma_\mu$ is the unique linear extension of $\mu$, i.e., $\mu = \sigma_\mu|_{\mc{P}(\cH_1) \times \mc{P}(\cH_2)}$.
    
    Conversely, it follows from the isomorphism in Eq.~(\ref{eq: CJI}), Choi's theorem \cite{Choi1975} and Def.~\ref{def: time-oriented states} that every quantum state restricts to a time-oriented measure $\mu: \mc{P}(\cH_1) \times \mc{P}(\cH_2) \ra [0,1]$.
\end{proof}

We finish this section with a few remarks.
We defined the canonical time orientation $\Psi$ on $\mc{L}(\cH) = (\mc{J}(\cH),\Psi)$ with respect to unitary symmetries. In turn, anti-unitary symmetries correspond with unitary symmetries on the algebra $\mc{L}_-(\cH)$ (and vice versa). Recalling that every anti-unitary operator is the product of a unitary and the time reversal operator further corroborates the notion of time-orientedness in Def.~\ref{def: time-oriented states}.

Moreover, note that Eq.~(\ref{eq: time-oriented states}) is invariant under changing time orientations, equivalently under time reversal in both subsystems (cf. \cite{Frembs2022b}).
On the other hand, it follows from Thm.~\ref{thm: main result} that applying time reversal to one subsystem only will generally map outside the set of bipartite quantum states. This is reminiscent of the positive partial transpose (PPT) criterion for bipartite entanglement due to Peres \cite{Peres1996,Horodeckisz1996}. In fact, following this analogy the criterion can be given a sharp physical interpretation in terms of time orientations \cite{Frembs2022a}.

Finally, note that Thm.~\ref{thm: main result} represents a generalisation of Gleason's theorem to composite systems. We extend this perspective to the general setting of von Neumann algebras in \cite{FrembsDoering2022b}.

\section{Conclusion and Outlook}\label{sec: Conclusion and outlook}

We studied the physical principle of no-disturbance on joint probability distributions over product observables, which
reduces to no-signalling as shown in Eq.~(\ref{eq: no-signalling from no-disturbance}). As such we identified no-disturbance as the key principle in \cite{BarnumEtAl2010},
which shows that correlations over product observables satisfying no-disturbance cannot exceed bipartite quantum correlations. However, no-disturbance is not sufficient to restrict to quantum states. Our main result, Thm.~\ref{thm: main result}, resolves this issue: by extending the scope of no-disturbance to dilations in Def.~\ref{def: no-disturbance for dilations}, and by enforcing a consistency condition with respect to unitary evolution in local subsystems in Def.~\ref{def: time-oriented states}, we show that the resulting joint probability distributions correspond with bipartite quantum states unambiguously.

Our research naturally relates to other approaches, which seek to single out quantum correlations from more general non-signalling distributions, by imposing additional principles, e.g. \cite{PawlowskiEtAl2009, vanDam2013,FritzEtAl2013,Popescu2014}. 
Apart from the extension of the no-disturbance principle to dilated systems, which fits nicely with the usual interpretation of mixed states as states of incomplete information as argued above, our only additional assumption already has a clear physical meaning in terms of the arrow of time \cite{Frembs2022a}.\\ 

\paragraph*{Acknowledgements.} This work is supported through a studentship in the Centre for Doctoral Training on Controlled Quantum Dynamics at Imperial College funded by the EPSRC, by grant number FQXi-RFP-1807 from the Foundational Questions Institute and Fetzer Franklin Fund, a donor advised fund of Silicon Valley Community Foundation, and ARC Future Fellowship FT180100317.

\clearpage
\bibliographystyle{siam}
\bibliography{bibliography}
\end{document}